\documentclass[12pt,regno]{amsart}
\usepackage{hyperref}
\usepackage{latexsym}
\usepackage{amssymb,amsfonts,amsmath}
\usepackage{graphicx}
\usepackage{indentfirst}
\usepackage{amsmath}
\usepackage{amsfonts}
\usepackage{amssymb}
\usepackage{comment}
\usepackage{mathtools}
\usepackage{amsthm}
\usepackage{enumerate}
\usepackage{extarrows}
\usepackage{relsize}

\usepackage[cmtip,all]{xy}
\usepackage{verbatim}
\usepackage{tikz-cd}
\usepackage{mathtools}
\usepackage{mathrsfs}
\usepackage{geometry}
\DeclareMathAlphabet{\mathpzc}{OT1}{pzc}{m}{it}

\newcommand*{\defeq}{\mathrel{\vcenter{\baselineskip0.5ex \lineskiplimit0pt
                     \hbox{\scriptsize.}\hbox{\scriptsize.}}}%
                     =}

\ifx\pdfoutput\undefined
\else
\fi
\hypersetup{colorlinks=false,bookmarksopen,bookmarksnumbered,citecolor=blue,
pdfstartview=FitH}

\parskip=\medskipamount
\newcommand{\virgolette}{``}

\theoremstyle{definition}

\theoremstyle{plain}
\newtheorem{thm}{Theorem}
\newtheorem{prop}[thm]{Proposition}
\theoremstyle{remark}
\newtheorem*{remark}{Remark}

\begin{document}

\begin{flushright}
\par\end{flushright}
\vspace{1cm}

\title{The Geometry of Picture Changing Operators}
\date{}

\author{C.\ A.\ Cremonini}

\email{carlo.alberto.cremonini@gmail.com}

\address{Faculty of Mathematics and Physics, Mathematical Institute, Charles University Prague, Sokolovska 49/83, 186 75 Prague}

\begin{abstract}

This note aims at clarifying some mathematical aspects of what is known in Physics as \emph{Picture Changing Operator} (PCO). In particular, we want to show that PCOs are cochain maps between the complex of differential forms (or superforms) and the complex of integral forms on a given supermanifold. We comment on the construction of (super)symmetric PCOs in terms of chain homotopies and we provide some physically relevant examples of applications.

\end{abstract}

\maketitle
\setcounter{footnote}{0}

\section*{Introduction}

The main goal of this short note is to clarify the mathematical meaning of what is known in Physics as \emph{Picture Changing Operator} (henceforth PCO). PCOs were introduced for the first time in \cite{FMS} as a device to shift from a given ghost representation to another, in the context of superconformal field theory, each of them leading to a complex with the operator given by the BRST differential. This device was essential in order to define a meaningful expression when evaluating string amplitudes via integration of correlators. These operators were first introduced via \virgolette bosonisation'', i.e., when fermions are described in terms of a chiral boson. Later, it was realised that such a bosonised fermion could be interpreted as a formal Dirac delta distribution of the bosonic ghost field. These operators were introduced in Supergeometry in \cite{Belopolsky}, where the translation from physical to mathematical apparatus is explained. In particular, in this paper, the author presents the realisation of integral forms in terms of Dirac deltas of (even) differential forms $d \theta$, in analogy to the interpretation of the bosonised ghost described above, and in this context, geometric PCOs are presented. This realisation of integral forms and PCOs, already established in superstring theory, was later reviewed in \cite{Witten} where the basics of Supergeometry, with an eye towards integration theory, are neatly presented. On the other hand, the rigorous aspects of supermanifolds, integral forms, integration theory etc., were already well-founded in the mathematical community. We refer to the book \cite{Manin} for a quick yet deep introduction to Supergeometry, or to the more recent \cite{Noja} for an up-to-date review of the subject, with a vast bibliography of fundamental papers (see also \cite{CNR}).

In this note\footnote{We will try to follow the notations of \cite{Manin, Witten, Noja}.}, we want to build a bridge between the mathematical and physical worlds by introducing the notion of Picture Changing Operators in a rigorous way. The idea is inherited from Physics: PCOs are used to go from one ghost vacuum to another and then from a realisation of physical states to another, or, analogously, from a BRST complex to another. This means that PCOs can be introduced as cochain maps between complexes, in this case, between the complex of differential forms and the complex of integral forms. We will explicitly construct the cochain map associated with the \virgolette space-time'' PCO and relate it to the (canonical) embedding of the reduced submanifold into the ambient supermanifold. We will then show that this map defines a quasi-isomorphism, i.e., an isomorphism when restricted to cohomology. This, in particular, gives an explicit, and very intuitive, realisation of the well-known isomorphism between superform and integral form cohomologies. We will then give an interpretation of different PCOs appearing in literature, e.g., respecting some (super)symmetries, in terms of chain homotopies. These topics constitute the core of this paper. Finally, we provide some explicit examples and some applications of physical interest. We conclude with a glimpse at future developments, focused on the rigorous definition of \virgolette pseudoforms'' (see again \cite{Witten}) which could lead, in the opinion of this author, to interesting applications both in Physics and in Mathematics.

\section{Picture Changing Operators}

This section is the core of these notes and is devoted to the constructive introduction of PCOs and their properties. In particular, given a supermanifold $\mathcal{SM}$, we want to show that PCOs are cochain maps between the complex of differential forms (or superforms) and the complex of integral forms on $\mathcal{SM}$. We will prove the following statement that sets the correspondence between the physical notation and the mathematical definition of PCOs:
\begin{prop}
	Given a supermanifold $\mathcal{SM}$ with $\dim \mathcal{SM} = (m|n)$, a \emph{Picture Changing Operator} (henceforth PCO) is a cohomology class $\mathbb{Y}^{(0|n)} \in H^{(0|n)} \left( \mathcal{SM} , d_{dR} \right)$, in the distributional realisation of integral forms; a PCO allows to obtain an integral form out of a superform by $\wedge$ multiplication and formal relations of distributions. PCOs correspond to cochain maps $f^{(\bullet)} : \left( \Omega^\bullet \left( \mathcal{SM} \right) , d_{dR} \right) \to \left( \Sigma^\bullet \left( \mathcal{SM} \right) , \delta \right)$, where $\Sigma^\bullet \left( \mathcal{SM} \right) \defeq \mathpzc{B}er \left( \mathcal{SM} \right) \otimes_{\mathcal{O}_{\mathcal{SM}}} S^{m-\bullet} \Pi \mathcal{T} \mathcal{SM}$ are integral forms (as defined, e.g., in \cite{Manin, Noja}):
	\begin{center}
		\begin{tikzcd}[cells={nodes={minimum height=2em}}]
			0 \arrow[r] & \mathcal{O}_\mathcal{SM} \arrow[r,"d_{dR}"] \arrow[d,"f^{(0)}"] &  \Omega^{(1)} \left( \mathcal{SM} \right) \arrow[r,"d_{dR}"] \arrow[d,"f^{(1)}"] &  \ldots \arrow[r,"d_{dR}"] & \Omega^{(m)} \left( \mathcal{SM} \right) \arrow[r,"d_{dR}"] \arrow[d,"f^{(m)}"] & \ldots  \\
			\ldots \arrow[r] & \Sigma^{(0)} \left( \mathcal{SM} \right) \arrow[r,"\delta"] & \Sigma^{(1)} \left( \mathcal{SM} \right) \arrow[r,"\delta"] & \ldots \arrow[r,"\delta"] & \Sigma^{(m)} \left( \mathcal{SM} \right) \arrow[r,"\delta"] & 0
		\end{tikzcd}
	\end{center}
	In particular, we have that $\mathbb{Y}^{(0|n)} \in H^{(0|n)} \left( \mathcal{SM} , d_{dR} \right)$ corresponds to the fact that $f^{(\bullet)}$ is a quasi-isomorphism. Finally, changing the representative of the class $\mathbb{Y}^{(0|n)}$ by adding $d_{dR}$-exact terms, corresponds to chain homotopies for $f^{(\bullet)}$.
\end{prop}
\begin{proof}
	We will try to structure the proof constructively and by steps, in order to have an exposition that interpolates between the mathematical and the physical languages.
	
	Let us assume that $\mathcal{SM}$ is connected (i.e., $\mathcal{SM}_{red}$ is connected), consider an open set $\mathcal{U}$ locally described by the set of coordinates $\left\lbrace x^i , \theta^\alpha \right\rbrace , i=1 , \ldots , m , \alpha = 1 , \ldots , n$; let $\mathbb{Y}^{(0|m)}_{s.t.}$, i.e., what in Physics literature is known as the \virgolette spacetime PCO'', be defined as
	\begin{equation}\label{NOPCOA}
		\mathbb{Y}_{s.t.}^{(0|n)} = \left( \prod_{\alpha = 1}^n \theta^\alpha \right) \left( \bigwedge_{\beta = 1}^n \delta \left( d \theta^\beta \right) \right) = \theta^1 \ldots \theta^n \delta \left( d \theta^1 \right) \wedge \ldots \wedge \delta \left( d \theta^n \right) \ .
	\end{equation}
	This $(0|n)$-form defines a cohomology class: the closure is easily verified thanks to the formal property $d \theta^\alpha \delta \left( d \theta^\alpha \right) = 0 , \forall \alpha = 1 , \ldots , n$; the non-exactness can be easily verified by checking that it does not exist an odd $(-1|n)$-form whose differential gives \eqref{NOPCOA}. Intuitively, since the de Rham differential can not produce $\theta$s, one has to start ab initio with all of them. The non-exactness can also be inferred by integration: one could integrate \eqref{NOPCOA} on a \virgolette fat point'', i.e., a point equipped with the given Grassmann algebra, and immediately see that the integral is non-zero.
	
	We can show that \eqref{NOPCOA} is globally defined: when patching from an open set to another, consider the general transformation
	\begin{eqnarray}
		\nonumber x^i &\mapsto& x'^i = f^i \left( x , \theta \right) = f^i_0 \left( x \right) + O \left( \theta^2 \right) \ , \\
		\label{NOPCOB} \theta^\alpha &\mapsto& \theta'^\alpha = \chi^\alpha \left( x , \theta \right) = \sum_{\beta = 1}^n \chi^{\alpha}_{\beta} \left( x \right) \theta^\beta + O \left( \theta^3 \right) \ ,
	\end{eqnarray}
	where $f^i$ and $\chi^\alpha$ are even and odd, respectively (we are omitting the explicit expressions of some terms in the general coordinate transformation as they are inessential for our purposes, because of the presence of $\left( \prod_{\alpha = 1}^n \theta^\alpha \right)$ in \eqref{NOPCOA}). The corresponding transformation for 1-forms reads
	\begin{eqnarray}
		\nonumber d x^i &\mapsto& d x'^i = \sum_{j=1}^m \left( \partial_j f^i_0 \left( x \right) \right) d x^j + O \left( \theta \right) \\
		\label{NOPCOC} d \theta^\alpha &\mapsto& d \theta'^\alpha = \sum_{\beta = 1}^n \left[ \left( d \chi^{\alpha}_{\beta} \left( x \right) \right) \theta^\beta + \chi^\alpha_\beta \left( x \right) d \theta^\beta \right] + O \left( \theta^2 \right) \ .
	\end{eqnarray}
	The transformed PCO reads
	\begin{equation}\label{NOPCOD}
		\mathbb{Y}_{s.t.}^{(0|n)} \mapsto \mathbb{Y}_{s.t.}^{\prime (0|n)} = \left( \prod_{\alpha = 1}^n \theta'^\alpha \right) \left( \bigwedge_{\beta = 1}^n \delta \left( d \theta'^\beta \right) \right) = $$ $$ = \left( \prod_{\alpha = 1}^n \left[ \sum_{\beta = 1}^n \chi^{\alpha}_{\beta} \left( x \right) \theta^\beta + O \left( \theta^3 \right) \right] \right) \left( \bigwedge_{\gamma = 1}^n \delta \left( \sum_{\delta = 1}^n \left[ \left( d \chi^{\gamma}_{\delta} \left( x \right) \right) \theta^\delta + \chi^\gamma_\delta \left( x \right) d \theta^\delta \right] + O \left( \theta^2 \right) \right) \right) = $$ $$ = \left( \prod_{\alpha = 1}^n \sum_{\beta = 1}^n \chi^{\alpha}_{\beta} \left( x \right) \theta^\beta \right) \left( \bigwedge_{\gamma = 1}^n \delta \left( \sum_{\delta = 1}^n \left[ \chi^\gamma_\delta \left( x \right) d \theta^\delta \right] \right) \right) = \left( \det \left( \chi^\alpha_\beta \right) \prod_{\alpha = 1}^n \theta^\alpha \right) \left( \frac{1}{\det \left( \chi^\alpha_\beta \right)} \bigwedge_{\beta = 1}^n \delta \left( d \theta^\beta \right) \right) = $$ $$ = \mathbb{Y}_{s.t.}^{(0|n)} \ , 
	\end{equation}
	where we have used the antisymmetry of $\theta$s and the formal properties of $\delta \left( d \theta \right)$s. In particular, the appearance of the term $\left( \det \left( \chi^\alpha_\beta \right) \right)^{-1}$ follows from the property $\delta \left( \lambda d \theta \right) = \frac{1}{\lambda} \delta \left( d \theta \right)$, without the absolute value (see \cite{Witten}). Notice that one can drop the connectedness assumption on $\mathcal{SM}$: in this case, one can define one PCO for each connected component of $\mathcal{SM}$. This simply reflects the isomorphism $H^{(0|0)} \left( \mathcal{SM} \right) \cong H^{(0|n)} \left( \mathcal{SM} \right)$, where $\dim H^{(0|0)} \left( \mathcal{SM} \right)$ is the number of connected components of $\mathcal{SM}$.
	
	We now construct the cochain map corresponding to the spacetime PCO. First of all, we have $f^{(i)} \defeq 0 , \forall i < 0 , i > m$. Given an open set $\mathcal{U}$ locally described by the set of coordinates $\left\lbrace x^i , \theta^\alpha \right\rbrace , i=1 , \ldots , m , \alpha = 1 , \ldots , n$, consider the local set of generators of $\mathcal{T}_{\mathcal{U}} \mathcal{SM}$ to be given by $\left\lbrace \partial_{x^i} , \partial_{\theta^\alpha} \right\rbrace , i = 1 , \ldots , m , \alpha = 1 , \ldots , n$ and the (parity-changed) dual basis $\left\lbrace dx^i , d \theta^\alpha \right\rbrace , i = 1 , \ldots , m , \alpha = 1 , \ldots , n$ spanning $\Pi \mathcal{T}_{\mathcal{U}}^* \mathcal{SM}$. A section of the Berezinian is generated by (see \cite{Manin, Noja})
	\begin{equation}\label{NOPCOE}
		\mathcal{D} \defeq \cdot \left[ dx^1 \wedge \ldots \wedge dx^m \otimes \partial_{\theta^1} \wedge \ldots \wedge \partial_{\theta^n} \right] \ .
	\end{equation}
	The first (non-zero) map $f^{(0)}$ is defined as
	\begin{eqnarray}
		\nonumber f^{(0)}: \mathcal{O}_{\mathcal{U}} &\to& \Sigma^{(0)} \left( \mathcal{U} \right) \\
		\label{NOPCOF} g \left( x , \theta \right) &\mapsto& f^{(0)} \left( g \left( x , \theta \right) \right) \defeq g \left( x , \theta \right) \prod_{\alpha = 1}^n \theta^\alpha \mathcal{D} \otimes \bigwedge_{i=1}^m \pi \partial_{x^i} = g_0 \left( x \right) \prod_{\alpha = 1}^n \theta^\alpha \mathcal{D} \otimes \bigwedge_{i=1}^m \pi \partial_{x^i} \ ,
	\end{eqnarray}
	where $g_0 \left( x \right) \defeq \left[ g \left( x , \theta \right) \right]_{\theta = 0}$ is the function obtained from $g$ by modding out the ideal generated by nilpotent coordinates and where $\pi$ denotes a parity change. This map is built as inspired by the short exact sequence defining the embedding $i: \mathcal{SM}_{red} \hookrightarrow \mathcal{SM}$, which at the level of sheaves reads
	\begin{equation}\label{NOPCOG}
		\begin{tikzcd}
			0\arrow{r} & \mathcal{J} \arrow{r} & \mathcal{O}_{\mathcal{SM}} \arrow{r}{i^*} & \mathcal{O}_{\mathcal{SM}_{red}} \arrow{r} & 0 \ \ .
		\end{tikzcd}
	\end{equation}
	Locally (in any open set $\mathcal{U}$), it implies $\mathcal{O}_{\mathcal{U}} = \mathcal{O}_{\mathcal{U}_{red}} \oplus \mathcal{J}$, where $\mathcal{J}$ is the ideal generated by the nilpotents and $\mathcal{O}_{\mathcal{U}_{red}}$ is obtained from $\mathcal{O}_{\mathcal{SM}_{red}}$, i.e., the reduced sheaf $\mathcal{O}_{\mathcal{SM}} / \mathcal{J}$. More specifically, for each open set $\mathcal{U}$, the map $f^{(0)}$ takes $\mathcal{O}_{\mathcal{U}}$ to $\mathcal{O}_{\mathcal{U}_{red}} \cdot \prod_{\alpha = 1}^n \theta^\alpha \mathcal{D} \otimes \bigwedge_{i=1}^m \pi \partial_{x^i}$. We can extend this to the whole $\mathcal{O}_{\mathcal{SM}}$ by patching and by noting that the section $\prod_{\alpha = 1}^n \theta^\alpha \mathcal{D} \otimes \bigwedge_{i=1}^m \pi \partial_{x^i}$ is global. Really, one has that under a general coordinate transformation \eqref{NOPCOB}, this section becomes
	\begin{equation}\label{NOPCOH}
		\prod_{\alpha = 1}^n \theta'^\alpha \mathcal{D'} \otimes \bigwedge_{i=1}^m \pi \partial_{x'^i} = \det \left( \chi^\alpha_\beta \right) \prod_{\alpha = 1}^n \theta^\alpha Ber \left( J \right) \mathcal{D} \otimes \left[ \det \left( \partial_j f^i_0 \right) \right]^{-1} \bigwedge_{i=1}^m \pi \partial_{x^i} = $$ $$ = \det \left( \chi^\alpha_\beta \right) \prod_{\alpha = 1}^n \theta^\alpha \frac{\det \left( \partial_j f^i_0 \right)}{\det \left( \chi^\alpha_\beta \right)} \mathcal{D} \otimes \left[ \det \left( \partial_j f^i_0 \right) \right]^{-1} \bigwedge_{i=1}^m \pi \partial_{x^i} = \prod_{\alpha = 1}^n \theta^\alpha \mathcal{D} \otimes \bigwedge_{i=1}^m \pi \partial_{x^i} \ ,
	\end{equation}
	where $J$ is the Jacobian of the transformation \eqref{NOPCOB} and where all the terms containing $\theta$s in the determinants have been dropped because of the presence of $\prod_{\alpha = 1}^n \theta^\alpha$.
	\begin{remark}
		The explicit form \eqref{NOPCOF} is related to \eqref{NOPCOA} in the following way: in the distributional realisation, \eqref{NOPCOE} can be written as $\displaystyle \mathcal{D} \defeq dx^1 \wedge \ldots \wedge dx^m \wedge \delta \left( d \theta^1 \right) \wedge \ldots \wedge \delta \left( d \theta^n \right)$; one can then obtain \eqref{NOPCOA} by acting with the contractions $\iota_{\partial_{x_1}} \ldots \iota_{\partial_{x_m}}$ (identified with $\displaystyle \bigwedge_{i=1}^m \pi \partial_{x^i}$) and multiplying by $\displaystyle \prod_{\alpha = 1}^n \theta^\alpha$.
	\end{remark}
	\begin{remark}
		If one works in the smooth category, the sequence \eqref{NOPCOG} splits non-canonically (smooth supermanifolds are split, see \cite{Batchelor}), i.e., one can choose a projection $p: \mathcal{SM} \to \mathcal{SM}_{red}$ or, on maps, $p^* : \mathcal{O}_{\mathcal{SM}_{red}} \to \mathcal{O}_{\mathcal{SM}}$ such that $\mathcal{O}_{\mathcal{SM}} = \mathcal{O}_{\mathcal{SM}_b} \oplus \mathcal{J}$ globally. At the level of transition functions, this means that one can choose \eqref{NOPCOB} such that the terms at least quadratic in $\theta$s are set to zero. Here, we want to emphasise that this assumption is not needed to show that the given cochain map, or, analogously, the spacetime PCO, is globally defined in the non-split cases as well.
	\end{remark}
	
	 The other maps are defined analogously; in particular, given the forms $\Omega^i \left( \mathcal{SM} \right)$, \eqref{NOPCOG} implies
	 \begin{equation}\label{NOPCOI}
	 	\begin{tikzcd}
			0\arrow{r} & Ker^{(p|0)} \left( i^* \right) \arrow{r} & \Omega^{(p|0)} \left( \mathcal{SM} \right) \arrow{r}{i^*} & \Omega^{(p)} \left( \mathcal{SM}_{red} \right) \arrow{r} & 0 \ \ .
		\end{tikzcd}
	 \end{equation}
	 Locally, this means that on an open set $\mathcal{U}$, we have $\Omega^{(p|0)} \left( \mathcal{U} \right) = \Omega^{(p)} \left( \mathcal{U}_{red} \right) \oplus Ker^{(p|0)} \left( i^* \right)$, where $\Omega^{(p)} \left( \mathcal{U}_{red} \right)$ represents forms of the reduced manifold (in particular, written in terms of $x$ and $dx$ only) while $Ker^{(p|0)} \left( i^* \right)$ represents forms at least linear in $\theta$ or $d \theta$. The map $f^{(p)}$ is then defined as
	 \begin{eqnarray}
		\nonumber f^{(p)}: \Omega^{(p|0)} \left( \mathcal{U} \right) &\to& \Sigma^{(p)} \left( \mathcal{U} \right) \\
		\label{NOPCOJ} \sum_{i_1 \ldots i_p = 1}^m g_{i_1 \ldots i_p} \left( x \right) dx^{i_1} \wedge \ldots \wedge dx^{i_p} + \eta &\mapsto& \sum_{i_1 \ldots i_p = 1}^m g_{i_1 \ldots i_p} \left( x \right) dx^{i_1} \wedge \ldots \wedge dx^{i_p} \cdot \prod_{\alpha = 1}^n \theta^\alpha \mathcal{D} \otimes \bigwedge_{i=1}^m \pi \partial_{x^i} \ ,
	\end{eqnarray}
	where $\eta \in Ker^{(p|0)} \left( i^* \right)$ and the \virgolette $\cdot$ '' notation means that we still have to use the pairing $d x^i \left( \pi \partial_{x_j} \right) = \delta^i_j$ to simplify \eqref{NOPCOF} as
	\begin{equation}\label{NOPCOK}
		\sum_{i_1 \ldots i_p = 1}^m g_{i_1 \ldots i_p} \left( x \right) dx^{i_1} \wedge \ldots \wedge dx^{i_p} \prod_{\alpha = 1}^n \theta^\alpha \mathcal{D} \otimes \bigwedge_{i=1}^m \pi \partial_{x^i} = $$ $$ = \frac{\left( -1 \right)^{pm}}{\left( m - p \right)!} \sum_{i_1 \ldots i_p = 1}^m \sum_{j_1 \ldots j_m = 1}^m  g_{i_1 \ldots i_p} \left( x \right) \delta^{i_1}_{j_1} \ldots \delta^{i_p}_{j_p} \prod_{\alpha = 1}^n \theta^\alpha \mathcal{D} \otimes \epsilon^{j_1 \ldots j_p j_{p+1} j_m} \bigwedge_{s=p+1}^m \pi \partial_{x^{j_s}} = $$ $$ = \frac{\left( -1 \right)^{pm}}{\left( m - p \right)!} \sum_{j_1 \ldots j_m = 1}^m g_{j_1 \ldots j_p} \left( x \right) \prod_{\alpha = 1}^n \theta^\alpha \mathcal{D} \otimes \epsilon^{j_1 \ldots j_p j_{p+1} j_m} \bigwedge_{s=p+1}^m \pi \partial_{x^{j_s}} \ ,
	\end{equation}
	where we have written $\displaystyle \bigwedge_{i=1}^m \pi \partial_{x^i} = \frac{1}{m!} \sum_{j_1 \ldots j_m = 1}^m \epsilon^{j_1 \ldots j_m} \pi \partial_{x_{j_1}} \wedge \ldots \wedge \pi \partial_{x_{j_m}}$. As in the previous cases, the global extension can be done by glueing different patches and recalling that $\prod_{\alpha = 1}^n \theta^\alpha \mathcal{D} \otimes \bigwedge_{i=1}^n \pi \partial_{x^i}$ is globally defined as shown in (4).
	
	To verify that \eqref{NOPCOJ} defines a cochain map, we have to check that
	\begin{equation}\label{NOPCOL}
		f^{(p+1)} \circ d = \delta \circ f^{(p)} \ , \ \forall p \in \mathbb{Z} \ ,
	\end{equation}
	where $\delta$ is the nilpotent operator of the integral form complex $\left( \Sigma^\bullet \left( \mathcal{SM} \right) , \delta \right)$ and is defined as (see, e.g., \cite{Manin, Noja})
	\begin{equation}\label{NOPCOLA}
		\delta \defeq \sum_{i=1}^m \mathcal{L}_{\partial_{x^i}} \otimes \frac{\partial}{\partial \left( \pi \partial_{x^i} \right)} - \sum_{\alpha=1}^n \mathcal{L}_{\partial_{\theta^\alpha}} \otimes \frac{\partial}{\partial \left( \pi \partial_{\theta^\alpha} \right)} \ .
	\end{equation}
	We have that \eqref{NOPCOL} holds trivially for $p < 0$ and $p > m$, as in these cases the maps $f^{(p)}$ are identically 0. For $0 \leq p \leq m$, consider $\omega \in \Omega^p \left( U \right)$, locally written as in \eqref{NOPCOJ}. We have
	\begin{equation}\label{NOPCOM}
		d \omega = \sum_{i_1 \ldots i_p = 1}^m \sum_{j=1}^m \partial_j g_{i_1 \ldots i_p} \left( x \right) dx^j \wedge dx^{i_1} \wedge \ldots \wedge dx^{i_p} + d \eta \ .
	\end{equation}
	Notice that $\left[ d Ker^{(p|0)} \left( i^* \right) \right] \in Ker^{(p+1|0)} \left( i^* \right)$. This leads to
	\begin{equation}\label{NOPCON}
		f^{(p+1)} \circ d \omega = \frac{\left( -1 \right)^{(p+1)m}}{\left( m - p - 1 \right)!} \sum_{j_1 \ldots j_m = 1}^m \partial_{j_1} g_{j_2 \ldots j_{p+1}} \left( x \right) \prod_{\alpha = 1}^n \theta^\alpha \mathcal{D} \otimes \epsilon^{j_1 \ldots j_{p+1} j_{p+2} j_n} \bigwedge_{s=p+2}^n \pi \partial_{x^{j_s}} \ .
	\end{equation}
	On the other hand, we have
	\begin{eqnarray}
		\nonumber \delta \circ f^{(p)} \omega &=& \delta \left[ \frac{\left( -1 \right)^{pm}}{\left( m - p \right)!} \sum_{j_1 \ldots j_m = 1}^m g_{j_1 \ldots j_p} \left( x \right) \prod_{\alpha = 1}^n \theta^\alpha \mathcal{D} \otimes \epsilon^{j_1 \ldots j_p j_{p+1} j_n} \bigwedge_{s=p+1}^n \pi \partial_{x^{j_s}} \right] = \\
		\label{NOPCOO} &=& \frac{\left( -1 \right)^{pm}}{\left( m - p \right)!} \left( -1 \right)^m (m-p) \sum_{j_1 \ldots j_m = 1}^m \partial_j g_{j_1 \ldots j_p} \left( x \right) \prod_{\alpha = 1}^n \theta^\alpha \mathcal{D} \otimes \epsilon^{j_1 \ldots j_p j_{p+1} j_n} \delta^j_{j_{p+1}} \bigwedge_{s=p+2}^n \pi \partial_{x^{j_s}} \ .
	\end{eqnarray}
	After renaming the indices, we see that \eqref{NOPCON} and \eqref{NOPCOO} coincide.
	
	Let us now verify that these maps define an isomorphism in cohomology: from the very definition of cochain map, we have that cocycles and coboundaries are respected. In particular, $f^{(\bullet)}$ induce linear morphisms for the cohomology groups
	\begin{equation}\label{NOPCOP}
		H (f)^{(p)}: H^p_{d_{dR}} \to H^p_\delta \ ,
	\end{equation}
	where with $H^p_{d_{dR}}$ and $H^p_\delta$ we denote the $p$-th cohomology groups of the superform and integral form complexes, respectively. In particular, we use the well-known isomorphism between the de Rham cohomology of a supermanifold and its reduced manifold:
	\begin{equation}\label{NOPCOQ}
		H^p \left( \mathcal{SM} , d_{dR} \right) \cong H^p \left( \mathcal{SM}_{red} , d_{dR} \right) \ , \ \forall 0 \leq p \leq m \ .
	\end{equation}
	This means that fixed $p$, there are no cohomology classes which can be represented as forms in $Ker^{(p|0)} \left( i^* \right)$ since they would be projectable to zero by reduction on the base manifold. On the other hand, $Ker^{(p|0)} \left( i^* \right)$ is by construction the kernel of $f^{(p)}$. These two facts imply that
	\begin{equation}\label{NOPCOR}
		\dim H^p \left( \mathcal{SM} , d_{dR} \right) = \dim H (f)^{(p)} \left[ H^p \left( \mathcal{SM} , d_{dR} \right) \right] \leq \dim H^p \left( \mathcal{SM} , \delta \right) \ .
	\end{equation}
	By recalling that $H^p \left( \mathcal{SM} , d_{dR} \right) \cong H^p \left( \mathcal{SM} , \delta \right) \forall p$, it follows by linearity that $H (f)^{(p)}$, i.e., the induced maps in cohomology, are isomorphisms $\forall p \in \mathbb{Z}$.
	
	Let us now move to the final part of the proposition: adding $d$-exact terms to $\mathbb{Y}^{(0|n)}$ corresponds to chain homotopies for $f^{(\bullet)}$. On the \virgolette physical notation'' side, one considers instead of the spacetime PCO \eqref{NOPCOA}, a different representative of $H^{(0|n)} \left( \mathcal{SM} , d_{dR} \right)$, by adding $d$-exact pieces:
	\begin{equation}\label{NOPCOS}
		\mathbb{Y}^{(0|n)}_{new} = \mathbb{Y}_{s.t.}^{(0|n)} + d \Lambda^{(-1|n)} \ .
	\end{equation}		
	This is very useful for many reasons; in particular, one can add exact terms such that the resulting PCO manifestly preserves some (super)symmetries. We will give a very simple example at the proof's end. In terms of cochain maps, two cochain maps $f^{(\bullet)} , e^{(\bullet)} : \left( \Omega^\bullet \left( \mathcal{SM} \right) , d_{dR} \right) \to \left( \Sigma^\bullet \left( \mathcal{SM} \right) , \delta \right)$ are chain homotopic if there exist a sequence of maps $h^{(\bullet)} : \left( \Omega^\bullet \left( \mathcal{SM} \right) , d_{dR} \right) \to \left( \Sigma^{\bullet -1} \left( \mathcal{SM} \right) , \delta \right)$ such that
	\begin{equation}\label{NOPCOT}
		e^{(p)} - f^{(p)} = \delta \circ h^{(p)} + h^{(p+1)} \circ d \ .
	\end{equation}
	In our case, we have that the maps $f^{(p)}$ read as in \eqref{NOPCOJ}, while the maps $e^{(p)}$ are defined as follows: from \eqref{NOPCOS}, we extract the expression for the integral form (in the physical notation) $\Lambda^{(-1|n)}$. By following the \virgolette dictionary'' between the physical and mathematical notations, we consider the corresponding integral form denoted as $h^{(-1)} \in \Sigma^{(-1)} \left( \mathcal{SM} \right)$. We can now consider the differential $\delta$ acting on $h^{(-1)}$, so that $\delta h^{(-1)} \in \Sigma^{(0)} \left( \mathcal{SM} \right)$. This object defines a sequence of maps
	\begin{eqnarray}
		\nonumber \left( \delta h \right)^{(p)} : \Omega^{(p)} \left( \mathcal{SM} \right) &\to& \Sigma^{(p)} \left( \mathcal{SM} \right) , \forall p \in \mathbb{Z} \\
		\label{NOPCOU} \omega &\mapsto& \left( \delta h \right)^{(p)} \left( \omega \right) \defeq \delta h^{(-1)} \circ \omega = \omega \cdot \delta h^{(-1)} \ ,
	\end{eqnarray}
	where, again, \virgolette $\cdot$'' stands for the pairing of the generators in $\omega$ and those in $\delta h^{(-1)}$. It is not difficult to show that $\left( \delta h \right)^{(p)}$ satisfy \eqref{NOPCOL} $\forall p \in \mathbb{Z}$, i.e., it defines a cochain map. The maps $e^{(p)}$ are then defined as $e^{(p)} \defeq f^{(p)} + \left( \delta h \right)^{(p)} , \forall p \in \mathbb{Z}$. The proof is easily concluded thanks to the fact that
	\begin{equation}\label{NOPCOV}
		\delta \left[ h^{(p)} \left( \omega \right) \right] = \left( \delta h^{(p)} \right) \left( \omega \right) + (-1)^{|h^{(p)}|} h^{(p+1)} \left( d \omega \right) \ ,
	\end{equation}
	as one can directly verify by using the definition of the two operators $\delta$ and $d$. Notice that $\displaystyle |h^{(p)}| = 1$ as follows from $\displaystyle |\delta h^{(p)}| = |f^{(p)}| = 0$ and from $\left| \delta \right| = 1$.

\end{proof}

\subsection{The quasi-inverse operator}

In the previous proposition, we have shown that the cochain map corresponding to a PCO is a quasi-isomorphism between the complexes of superforms and integral forms. This means that we should be able to construct a cochain map from the complex of integral forms to the complex of superforms $Z^{(i)}: \Sigma^{(i)} \left( \mathcal{SM} \right) \to \Omega^{(i)} \left( \mathcal{SM} \right)$ which when restricted to cohomology, is the inverse of the cochain map corresponding to a PCO. This is well known in Physics\footnote{To be honest, to our knowledge, the operator we refer to as \virgolette inverse operator'' was introduced in Physics before the other.} and an explicit form of this operator, named $Z$, is given in terms of formal Heaviside step-function of the contraction along basis odd vector fields:
\begin{equation}\label{TQIOA}
	Z = \circ_{\alpha = 1}^n Z_\alpha \ , \ Z_\alpha \defeq -i \left[ d , \Theta \left( \iota_{\partial_{\theta^\alpha}} \right) \right] \ ,
\end{equation}
where the action of the Heaviside step-function is understood through its Fourier representation (see \cite{Belopolsky, Witten, Crem-Gras} for more details and examples). This operator and \eqref{NOPCOA} satisfy the following quasi-inverse identities (the first on integral forms, the second on superforms):
\begin{equation}\label{TQIOB}
	Z \mathbb{Y} Z = Z \ , \ \mathbb{Y} Z \mathbb{Y} = \mathbb{Y} \ .
\end{equation}
The direct translation from the compact expression \eqref{TQIOA} to a cochain map is not immediate, but it is not difficult to infer it once one has seen how \eqref{TQIOA} acts in the distributional realisation (see, e.g., \cite{Crem-Gras}). In particular, we have that $Z^{(i)} \defeq 0 , \forall i > m, i<0$. For $i=m$, on a given open set $\mathcal{U}$, one has
\begin{eqnarray}
	\nonumber Z^{(m)}: \Sigma^{(m)} \left( \mathcal{U} \right) &\to& \Omega^{(m)} \left( \mathcal{U} \right) \\
	\label{TQIOC} g \left( x , \theta \right) \mathcal{D} &\mapsto& Z^{(m)} \left( g \left( x , \theta \right) \mathcal{D} \right) \defeq g_{top} ( x ) dx^1 \wedge \ldots \wedge dx^m \ , 
\end{eqnarray}
where $g_{top} \left( x \right)$ is the component of $g (x , \theta)$ with all the odd variables: $g \left( x , \theta \right) = g_0 \left( x \right) + \ldots + g_{top} \left( x \right) \theta^1 \ldots \theta^n$. On a general integral form in $\Sigma^{(m-p)} \left( \mathcal{U} \right)$, we have $\forall 0 \leq p \leq m$
\begin{eqnarray}
	\nonumber Z^{(m-p)}: \Sigma^{(m-p)} \left( \mathcal{U} \right) &\to& \Omega^{(m-p)} \left( \mathcal{U} \right) \\
	\label{TQIOD} \sum_{i_1 , \ldots , i_{p} = 1}^m g_{i_1 \ldots i_{p}} \left( x , \theta \right) \mathcal{D} \otimes \bigwedge_{s = 1}^p \pi \partial_{x^s} + \eta &\mapsto& \sum_{i_1 , \ldots , i_{p} = 1}^m g_{i_1 \ldots i_{p}, top} \left( x \right) \iota_{\partial_{x^{i_1}}} \ldots \iota_{\partial_{x^{i_p}}} \left( dx^1 \wedge \ldots \wedge dx^m \right) \ ,
\end{eqnarray}
where $\eta$ is the part of the integral form containing at least one (parity changed) odd vector field $\pi \partial_{\theta}$. Equation \eqref{TQIOD} could be further simplified by noting that
\begin{equation*}
	\iota_{\partial_{x^{i_1}}} \ldots \iota_{\partial_{x^{i_p}}} \left( dx^1 \wedge \ldots \wedge dx^m \right) = \sum_{j_1 , \ldots , j_m = 1}^m \iota_{\partial_{x^{i_1}}} \ldots \iota_{\partial_{x^{i_p}}} \frac{1}{n!} \epsilon_{j_1 \ldots j_m} dx^{j_1} \wedge \ldots \wedge dx^{j_m} = $$ $$ = \frac{1}{\left( m - p \right)!} \delta^{j_1}_{i_1} \ldots \delta ^{j_p}_{i_p} \epsilon_{j_1 \ldots j_m}  dx^{j_{p+1}} \wedge \ldots \wedge dx^{j_m} \ .
\end{equation*}
It is very easy to see from the definitions \eqref{NOPCOJ} and \eqref{TQIOC} that these maps satisfy
\begin{equation}\label{TQIOE}
	f^{(p)} \circ Z^{(p)} \circ f^{(p)} = f^{(p)} \ , \ Z^{(p)} \circ f^{(p)} \circ Z^{(p)} = Z^{(p)} \ , \ \forall p \in \mathbb{Z} \ ,
\end{equation}
when acting on superforms and integral forms, respectively. Analogously, one can immediately see that the cochain map \eqref{TQIOC} is the inverse of the cochain map \eqref{NOPCOJ} when restricted to cohomology.

\begin{remark}
	When proving that \eqref{NOPCOJ} defines a quasi-isomorphism between the superform and integral form cochain complexes, we explicitly used the well-known result $H^p \left( \mathcal{SM} , d_{dR} \right) \cong H^p \left( \mathcal{SM} , \delta \right) \forall p$. Here, we want to emphasise that we could have proven that \eqref{NOPCOJ} is a quasi-isomorphism without relying on that result, but by simply constructing its quasi-inverse \eqref{TQIOE}. This could be seen as an alternative, constructive (sketch of the) proof of the isomorphism of the cohomology of the two complexes.
\end{remark}

\subsection{Example: supersymmetric PCO}

We want to present here an easy example where we can introduce cochain maps corresponding to PCOs preserving important symmetries of the physical theory. Let us consider the coset manifold $\mathcal{SM}$ given by $\displaystyle \mathcal{SM} \defeq \frac{ISO \left( 1,2|1 \right)}{SO \left( 1,2 \right)}$, that is, the $d=3 , N=1$ flat superspace. In particular, we are going to construct a PCO which is manifestly supersymmetric.

The coset algebra $\displaystyle \frac{\mathfrak{iso} \left( 1,2|1 \right)}{\mathfrak{so} \left( 1,2 \right)}$ is generated by three even vectors $\left\lbrace P_a \right\rbrace_{a=0}^2$ in the vector representation of $\mathfrak{so}(1,2)$ and two odd vectors $\left\lbrace D_\alpha \right\rbrace_{\alpha=1}^2$ in the spinor representation of $\mathfrak{so}(1,2)$ (often referred to as \virgolette superderivatives''); the only non-trivial commutation relations read\footnote{We denote with $\gamma^a$ a basis of real, symmetric Dirac matrices: $\displaystyle \gamma^0 = - Id , \gamma^1 = \sigma^3 , \gamma^2 = - \sigma^1$, $\sigma^i$ being the Pauli matrices.} (we denote with $[ \cdot , \cdot ]$ the graded commutators)
\begin{equation}\label{ESPCOA}
	\left[ D_\alpha , D_\beta \right] = -2 \gamma^a_{\alpha \beta} P_a \ ,
\end{equation}
where the sum over repeated indices is understood. On the dual side, we can introduce the Maurer-Cartan forms $\left\lbrace V^a \right\rbrace_{a=0}^2$ odd and $\left\lbrace \psi^\alpha \right\rbrace_{\alpha = 1}^2$ even, satisfying the Maurer-Cartan equations
\begin{equation}\label{ESPCOB}
	d V^a = \psi^\alpha \gamma^a_{\alpha \beta} \psi^\beta \ , \ d \psi^\alpha = 0 \ .
\end{equation}
Given a set of local coordinates $\left\lbrace x^a \right\rbrace_{a=0}^3$ even and $\left\lbrace \theta^\alpha \right\rbrace_{\alpha = 1}^2$ odd, the holonomic basis of vector fields is given by $\left\lbrace \partial_a \equiv \frac{\partial}{\partial x^a} \right\rbrace_{a=0}^2$ and $\left\lbrace \partial_\alpha \equiv \frac{\partial}{\partial \theta^\alpha} \right\rbrace_{\alpha = 1}^2$. The anholonomous generators $P_a$ and $D_\alpha$ can be realised as
\begin{equation}\label{ESPCOC}
	P_a = \partial_a \ , \ D_\alpha = \partial_\alpha - \theta^\beta \gamma^a_{\alpha \beta} \partial_a \ .
\end{equation}
The cotangent space is generated by $\left\lbrace d x^a \right\rbrace_{a=0}^2$ and $\left\lbrace d \theta^\alpha \right\rbrace_{\alpha = 1}^2$ and the Maurer-Cartan forms then read
\begin{equation}\label{ESPCOD}
	V^a = dx^a + \theta^\alpha \gamma^a_{\alpha \beta} d \theta^\beta \ , \ \psi^\alpha = d \theta^\alpha \ .
\end{equation}
We can then introduce the supersymmetry generators\footnote{There are different conventions in the literature regarding the definition of supersymmetry generators, we follow \cite{Wess-Bagger}.} as
\begin{equation}\label{ESPCOE}
	Q_\alpha = \partial_\alpha + \theta^\beta \gamma^a_{\alpha \beta} \partial_a \ .
\end{equation}
One can immediately verify the following properties (we denote with $\mathcal{L}_X$ the Lie derivative along the vector field $X$):
\begin{equation}\label{ESPCOF}
	\left[ Q_\alpha , Q_\beta \right] = 2 \gamma^a_{\alpha \beta} P_a \ , \ \left[ Q_\alpha , D_\beta \right] = 0 \ , \ \mathcal{L}_{Q_\alpha} V^a = \left( d \iota_{Q_\alpha} - \iota_{Q_\alpha} d \right) V^a = 0 = \mathcal{L}_{Q_\alpha} \psi^\beta \ ;
\end{equation}
the first property states that two supersymmetries give a translation(i.e., one has a supertraslation algebra); the second property reflects the fact that the superderivarive of a superfield is again a superfield; the third and fourth properties state that the Maurer-Cartan basis of forms is manifestly supersymmetric (i.e., the Lie derivative of the basis generators along the odd vectors $Q_\alpha$ is zero). In particular, we immediately have that any object constructed out of the Maurer-Cartan forms only will be automatically supersymmetric, as follows from Leibniz's rule. Consider the $(0|2)$-integral form (the numeric factor is for later convenience)
\begin{equation}\label{ESPCOG}
	\mathbb{Y}^{(0|2)}_{susy} = - \frac{1}{8} V^a \wedge V^b \wedge \gamma^{c \alpha \beta} \epsilon_{abc} \iota_\alpha \iota_\beta \delta \left( \psi^1 \right) \wedge \delta \left( \psi^2 \right) \ ,
\end{equation}
we claim that it is a Picture Changing Operator, in the sense defined in the previous section. This integral form corresponds (up to a numeric factor) to the cochain map defined as
\begin{eqnarray}
	\nonumber e^{(p)} : \Omega^{(p|0)} \left( \mathcal{SM} \right) &\to& \Sigma^{(p)} \left( \mathcal{SM} \right) \\
	\label{ESPCOGA} \omega &\mapsto& \omega \cdot \mathcal{D} \otimes \sum_{a=0,1,2} \sum_{\alpha , \beta = 1,2} \pi P_a \wedge \pi D_\alpha \wedge \pi D_\beta \gamma^{a \alpha \beta} \ , \ \forall p \in \mathbb{Z} \ ,
\end{eqnarray}
where $\mathcal{D} \defeq \left[ V^1 \wedge V^2 \wedge V^3 \otimes D_1 \wedge D_2 \right]$. The easiest way to prove that \eqref{ESPCOG} is a Picture Changing Operator is to show that it differs from \eqref{NOPCOA} (in the case $n=2$) by an exact term and to do this we simply expand the Maurer-Cartan forms in the holonomic basis as in \eqref{ESPCOD} (we drop the \virgolette $\wedge$'' symbol):
\begin{equation}\label{ESPCOH}
	\mathbb{Y}^{(0|2)}_{susy} = - \frac{1}{8} \left( dx^a + \theta^\alpha \gamma^a_{\alpha \beta} d \theta^\beta \right) \left( dx^b + \theta^\gamma \gamma^b_{\gamma \delta} d \theta^\delta \right) \gamma^{c \epsilon \zeta} \epsilon_{abc} \iota_\epsilon \iota_\zeta \delta \left( d \theta^1 \right) \delta \left( d \theta^2 \right) = $$ $$ = \theta^1 \theta^2 \delta \left( d \theta^1 \right) \delta \left( d \theta^2 \right) + d \left[ \frac{1}{2} dx^0 \theta^1 \theta^2 \iota_2^2 \delta \left( d \theta^1 \right) \delta \left( d \theta^2 \right) +  \frac{1}{2} dx^0 \theta^1 \theta^2 \iota_1^2 \delta \left( d \theta^1 \right) \delta \left( d \theta^2 \right) + \right. $$ $$ \left. - \frac{1}{4} dx^0 dx^1 \theta^1 \iota_1^2 \iota_2 \delta \left( d \theta^1 \right) \delta \left( d \theta^2 \right) - \frac{1}{4} dx^1 dx^2 \theta^2 \iota_1^2 \iota_2 \delta \left( d \theta^1 \right) \delta \left( d \theta^2 \right) - \frac{1}{4} dx^1 dx^2 \theta^1 \iota_1 \iota_2^2 \delta \left( d \theta^1 \right) \delta \left( d \theta^2 \right) + \right. $$ $$ \left. + \frac{1}{4} dx^2 dx^0 \theta^2 \iota_1^2 \iota_2 \delta \left( d \theta^1 \right) \delta \left( d \theta^2 \right) - \frac{1}{4} dx^2 dx^0 \theta^1 \iota_1 \iota_2^2 \delta \left( d \theta^1 \right) \delta \left( d \theta^2 \right) \right] = $$ $$ = \mathbb{Y}^{(0|2)}_{s.t.} + d \Lambda^{(-1|2)} \ .
\end{equation}
From this expression, we can immediately read a chain homotopy between \eqref{ESPCOGA} (up to a scale factor) and \eqref{NOPCOA} in terms of $\Lambda^{(-1|2)}$:
\begin{eqnarray}\label{ESPCOI}
	\nonumber h^{(-1)} \left( \bullet \right) \defeq \bullet \cdot \mathcal{D} & \otimes & \left[ \theta^1 \theta^2 \pi \partial_{x^1} \wedge \pi \partial_{x^2} \otimes \left( \left( \pi \partial_{\theta^2} \right)^{\wedge 2} + \left( \pi \partial_{\theta^1} \right)^{\wedge 2} \right) -\frac{1}{2} \left( \theta^1 \pi \partial_{x^2} \otimes \left( \pi \partial_{\theta^1} \right)^{\wedge 2} \pi \partial_{\theta^2} + \right. \right. \\
	&& + \theta^2 \pi \partial_{x^0} \otimes \left( \left( \pi \partial_{\theta^1} \right)^{\wedge 2} \wedge \pi \partial_{\theta^2} + \pi \partial_{\theta^1} \wedge \left( \pi \partial_{\theta^2} \right)^{\wedge 2} \right) + \\
	\nonumber && \left. - \theta^2 \pi \partial_{x^1} \left( \pi \partial_{\theta^1} \right)^{\wedge 2} \wedge \pi \partial_{\theta^2} + \theta^1 \pi \partial_{x^1} \otimes \pi \partial_{\theta^1} \wedge \left( \pi \partial_{\theta^2} \right)^{\wedge 2} \right) \Big{]} \ .
\end{eqnarray}

The importance of constructing cochain maps that preserve some symmetries is in close contact with Physics. In the following, we comment on two possible fields of application.

\subsubsection*{Rheonomy, or Group Manifold approach}

Supergravity and Superstring (Field) Theory are physical frameworks where supergeometry finds its applications. There are several ways to construct Lagrangians for these theories, one of these is \emph{Rheonomy} (see the books \cite{CDF}). The procedure goes approximately as follows: one defines the Bianchi identities for the field content of the theory (the \virgolette supermultiplets'') living on a certain supermanifold $\mathcal{SM}$ (e.g., a supergroup or a supercoset when in vacuum) and constrain them to satisfy the \virgolette rheonomic constraints'', i.e., some compatibility conditions with the reduction on the base manifold. Out of these data, one can construct a Lagrangian as a form $\mathcal{L} \in \Omega^{(\dim_0 \mathcal{SM} | 0)}$. In order to construct an action, and in particular a consistent variational principle, one should produce a section of the Berezinian of $\mathcal{SM}$ out of $\mathcal{L}$ and this is done via a Picture Changing Operator $\mathbb{Y}^{(0|\dim_1 \mathcal{SM})} \in H^{(0|\dim_1 \mathcal{SM})} \left( \mathcal{SM} \right)$. This datum represents the Poincar\'e dual of the embedding of the reduced submanifold $\mathcal{SM}_{red} \hookrightarrow \mathcal{SM}$ and the choice of representative can tune the resulting action and degrees of freedom. In particular, if $d \mathcal{L} = 0$ (this condition is satisfied when the theory admits a finite number of auxiliary fields) and $\partial \mathcal{SM} = \emptyset$, then all the representatives will lead to the same action:
\begin{equation}
	S \left[ \mathcal{L} , \tilde{\mathbb{Y}}^{(0|\dim_1 \mathcal{SM})} = \mathbb{Y}^{(0|\dim_1 \mathcal{SM})} + d \Lambda^{(-1|\dim_1 \mathcal{SM})} \right] - S \left[ \mathcal{L} , \mathbb{Y}^{(0|\dim_1 \mathcal{SM})} \right] = $$ $$ = \int_{\mathcal{SM}} \mathcal{L} \wedge d \Lambda^{(-1|\dim_1 \mathcal{SM})} = \int_{\mathcal{SM}} (-1)^{\left| \mathcal{L} \right|} d \left( \mathcal{L} \wedge \Lambda^{(-1|\dim_1 \mathcal{SM})} \right) = 0 \ ,
\end{equation}
because of Stokes' theorem. In these cases, the choice of the Picture Changing Operator allows to have the desired manifest (super)symmetries and changing representative intertwines different realisations of the same theory. This shows that the same theory and the same degrees of freedom may be contained in multiple copies in the same superform Lagrangian, allowing the Picture Changing Operator to select different, but equivalent, subsectors of it. A typical example is (super)Chern-Simons theory, the whole analysis can be found in, e.g., \cite{Crem-Gras}.

If $d \mathcal{L} \neq 0$, or in the presence of boundaries, different Picture Changing Operators will give rise to different theories. In these cases, the choice of preserved (super)symmetries will determine the resulting actions' possible (super)symmetries. A recent example of this class of theories can be found, e.g., in \cite{M5brane}.

\subsubsection*{Wilson Operators, or (higher) holonomies}

(Super) Wilson Operators, a.k.a. (higher) holonomies, represent another field where the algebraic methods of Picture Changing Operators can be useful when dealing with invariances and, in particular, oddly-generated symmetries. Let us first recall the definition of these operators: let $\mathcal{SM}$ be a supermanifold and $\Sigma \hookrightarrow \mathcal{SM}$ a submanifold (a line, a surface, etc.), we define the Wilson Operator associated to a (eventually super or higher) connection $A$ and to $\Sigma$ as
\begin{equation}\label{WOOHHA}
	\mathcal{W} \defeq \text{Tr} P \exp \left( \int_\Sigma A^* \right) \ ,
\end{equation}
where $\text{Tr}$ denotes the trace on the gauge group, $P$ is an ordering on the path/surface/etc. and $A^*$ is the pullback of $A$ on $\Sigma$. By making use of the embeddings $\Sigma \hookrightarrow \mathcal{SM}_{red} \hookrightarrow \mathcal{SM}$ one can rewrite the exponent as the integral of the connection $A$ on the whole supermanifold $\mathcal{SM}$:
\begin{equation}\label{WOOHHB}
	\int_\Sigma A^* = \int_{\mathcal{SM}} A \wedge \mathbb{Y}_{\Sigma \hookrightarrow \mathcal{SM}_{red}} \wedge \mathbb{Y}_{\mathcal{SM}_{red} \hookrightarrow \mathcal{SM}} \ .
\end{equation}
The integral form $\mathbb{Y}_{\Sigma \hookrightarrow \mathcal{SM}_{red}} \wedge \mathbb{Y}_{\mathcal{SM}_{red} \hookrightarrow \mathcal{SM}}$ corresponds to the composition of the map defined by the Poincar\'e dual (in the usual sense, i.e., for manifolds) of the embedding $\Sigma \hookrightarrow \mathcal{SM}_{red}$ and the map describing the embedding $\mathcal{SM}_{red} \hookrightarrow \mathcal{SM}$ as described in the previous section. Since both maps are cochain maps, their composition is a cochain map relative to the embedding $\Sigma \hookrightarrow \mathcal{SM}$; as we discussed in the previous example, one can then play with chain homotopies in order to write cochain maps preserving certain symmetries. This is particularly useful as the symmetries of \eqref{WOOHHA} can be studied in terms of the symmetries of these cochain maps, the main advantage is the reduction of a difficult geometrical problem, to an easier algebraic one. This could be useful when studying the classification of \virgolette BPS operators'', i.e., operators that brake part of the supersymmetries, or for other subtle tasks such as the implementation of \virgolette kappa symmetry'' for these operators (this was shown for some specific examples in \cite{WL, WS}).

\subsubsection{A Comment on Symmetric PCOs}

The quest for finding some symmetry-preserving PCOs may be quite difficult, but there are some cases where one can simplify it as it is related to Chevalley-Eilenberg (co)homology. In particular, if one considers a semi-simple Lie algebra $\mathfrak{g}$ over a characteristic-zero field $\mathbb{K}$, then any cohomology class $\omega \in H^p \left( \mathfrak{g} , \mathbb{K} \right)$ contains a $\mathfrak{g}$-invariant cocycle (see the original paper \cite{Chevalley-Eilenberg} for the proof). This can be easily extended to Lie superalgebras, a broad treatment of Lie superalgebras cohomology can be found in the book \cite{Fuks}. Recently, a notion of integral forms and their cohomology in the setting of Lie superalgebras was introduced in \cite{CE} and \cite{CE2}: by mimicking their very definition in the geometric context, algebraic integral forms (with values in a given module $V$) are defined as chains tensored with the \virgolette Haar Berezinian'' $\mathcal{D}$ with values on $V$:
\begin{equation}\label{ACOSPCOA}
	C_\bullet \left( \mathfrak{g} , \mathcal{D} \otimes V \right) \equiv \Sigma^\bullet \left( \mathfrak{g} , V \right) \defeq \mathcal{D} \otimes S^{\dim_0 \mathfrak{g} - \bullet} \Pi \mathfrak{g} \ .
\end{equation}
If one considers $V = \mathbb{K}$ and if $\mathfrak{g}$ is unimodular (see \cite{CE, CE2} for more details), then the module $\mathcal{D}$ is trivial (in the sense that it is $\mathfrak{g}$-invariant) and the computation of the integral form cohomology coincides with the homology. Again, the theorem cited above holds and one can choose a $\mathfrak{g}$-invariant homology class to tensor with the Haar Berezinian; this means that the defined integral form is $\mathfrak{g}$-invariant.

If the Lie superalgebra $\mathfrak{g}$ is the local version of a given Lie supergroup $G$, the previous construction can be used to produce $\mathfrak{g}$-invariant forms for $G$; the algebra is generated by the left-invariant (for example) vector fields, whose duals are the Maurer-Cartan forms, and integral form built from left-invariant vector fields and the Haar Berezinian (which in the unimodular case is both left- and right-invariant) are automatically left-invariant. In the distributional realisation of integral forms, these correspond to distributions built in terms of the Maurer-Cartan forms only and, in particular, with no explicit dependence on the local coordinates. One could then use simple algebraic tools to construct PCOs compatible with the symmetries with the manifold under consideration. In the example above, the PCO in \eqref{ESPCOGA} was built from the homology class $\pi P_a \wedge \pi D_\alpha \wedge \pi D_\beta \gamma^{a \alpha \beta}$ and out of this, one is guaranteed the invariance under the action of the $Q$ generators. This procedure can be used in reverse to see whether a symmetric PCO can be built in terms of Maurer-Cartan forms only and when this is not the case; for example, in the case of $d=11 , N=1$ flat superspace, i.e., $\displaystyle \mathcal{SM} \defeq \frac{ISO \left( 1,10|1 \right)}{SO \left( 1,10 \right)}$, with associated coset Lie superalgebra $\displaystyle \mathfrak{g} \defeq \frac{\mathfrak{iso} \left( 1,10|1 \right)}{\mathfrak{so} \left( 1,10 \right)}$, there is no cohomology class $\omega \in \Sigma^0 \left( \mathfrak{g} \right) \equiv \Omega^{(0|32)} \left( \mathfrak{g} \right)$, hence we know a priori that it is not possible to introduce a PCO built out of Maurer-Cartan forms only.

The importance of this construction manifests widely in Physics and one famous instance is the construction of the Green-Schwarz formulation of superstring theory (see \cite{GSW}). The action is defined as a sigma model action plus a Wess-Zumino term:
\begin{equation}\label{ACOSPCOB}
	S = - \frac{1}{2} \int_{\Sigma} \left( V^\mu \right)^* \wedge \star \left( V^\nu \right)^* \eta_{\mu \nu} + S_2 = - \frac{1}{2} \int_{\Sigma} \sqrt{h} h^{ij} \Pi_i^\mu \Pi_j^\nu \eta_{\mu \nu} + S_2 \ , \ \Pi^\mu_i = \partial_i X^\mu - i \bar{\theta}^A \Gamma^\mu \partial_i \theta^A \ ,
\end{equation}
where $X^\mu : \Sigma \to \mathcal{SM}, \mu = 0 , \ldots , \dim_0 \mathcal{SM} - 1$ are maps from a two-dimensional manifold $\Sigma$ (the \virgolette worldsheet'') to a target supermanifold $\mathcal{SM}$, chosen to be super Minkowski with $N=0,1,2$ supercharges ($A = 0 , \ldots , N$ is an $R$-symmetry index), with $\dim_0 \mathcal{SM} = 3,4,6,10$. The second term of the action (which is needed for the action to be \virgolette kappa symmetric'', in order to have the matching of even-odd degrees of freedom, in other words, to be supersymmetric in the worldsheet) has a quite complicated expression when integrated on $\Sigma$ and it is difficult to check its supersymmetry property (in the target). Conversely, it is very easy to construct it in the way presented above: one can consider the Lie algebra cohomology\footnote{We are abusing the name since we are dealing with a coset (super Poincar\'e mod Lorentz), so we should talk about \emph{relative} Lie algebra cohomology.} of supertranslations $\mathfrak{g}$ in the target and immediately see that in the cases of dimensions under consideration there exist a three form $\Omega^{(3)} \left( \mathfrak{g} \right) \ni \omega = V^\mu \Gamma_\mu \psi^A \psi^A$ so that the Wees-Zumino term can be written as
\begin{equation}\label{ACOSPCOC}
	S_2 = c \int_\Lambda \omega^* \ , \ \partial \Lambda = \Sigma \ ,
\end{equation}
where with $\omega^*$ we denote the pull back of $\omega$ on a manifold $\Lambda$ whose boundary is the worldsheet $\Sigma$ and $c$ is a given coefficient. In this case, not only the supersymmetry of the second term of the Lagrangian is guaranteed by construction, but also the existence of this kind of term in certain dimensions is classified by Lie superalgebra cohomology.

\section{Future developments: Pseudoforms}

In this section, we set the basis for the ideas underlying future developments; we will be sketchy, leaving the rigorous definitions for a forthcoming analysis. In particular, we want to comment on \emph{pseudoforms} and PCOs related to them: in \cite{Witten}, the author reviews the distributional realisation of integral forms introduced as a module of the Clifford-Weyl algebra defined by differential forms and contractions on a given supermanifold (the complex of integral forms is introduced by acting with contractions, interpreted as \virgolette creation operators'', acting on the \virgolette ground state'' which is a form in the Berezinian sheaf, annihilated by any multiplication with a form). Analogously, he introduces superforms as an inequivalent module of this algebra (now the complex of superforms is introduced by acting via multiplications with 1-forms, interpreted as creation operators, on the ground state 1 annihilated by any contraction). Finally, a third class of modules is defined by introducing a ground state which is annihilated by some contractions and some form multiplications. These ground states define new complexes of forms, obtained by acting on the ground state with the \virgolette complementary'' contractions or form multiplications; the forms in these complexes are called \emph{pseudoforms} and they are realised in terms of forms with a non-maximal and non-zero number of Dirac deltas supported on the even forms $d \theta$s.

Pseudoforms are not yet fully understood from a geometric point of view, as, for example, they do not transform tensorially: let us consider the basic example of $\mathbb{R}^{(0|2)}$, we can write a pseudoform supported on $d \theta^1$ as
\begin{equation}
	\omega = \delta \left( d \theta^1 \right) \ .
\end{equation}
If we consider the transformation $d \theta^1 \mapsto d \theta^1 + d \theta^2$, we obtain a formal series as
\begin{equation}
	\delta \left( d \theta^1 \right) \mapsto \delta \left( d \theta^1 + d \theta^2 \right) = \sum_{p=0}^\infty \frac{\left( d \theta^2 \right)^p}{p!} \delta^{(p)} \left( d \theta^1 \right) \ ,
\end{equation}
where $\displaystyle \delta^{(p)} \left( d \theta \right) \equiv \left( \iota_{\frac{\partial}{\partial \theta}} \right)^p \delta \left( d \theta \right)$. Nevertheless, in the recent \cite{CE2}, pseudoforms have been introduced rigorously in a simplified context: Lie superalgebras. Given a Lie superalgebra $\mathfrak{g}, \dim \mathfrak{g} = (m|n)$ and a Lie sub-superalgebra $\mathfrak{h}, \dim \mathfrak{h} = (p|q)$, one can introduce two infinite-dimensional $\mathfrak{g}$-modules associated to the Berezinian $\mathfrak{h}$-modules of $\mathfrak{h}$ and $\mathfrak{k} = \mathfrak{g}/\mathfrak{h}$:
\begin{equation}
	V_{\mathfrak{h}}^{(p|q)} \defeq \bigoplus_{i=0}^\infty \left( S^i \Pi \mathfrak{h} \otimes \mathpzc{B}er \left( \mathfrak{h} \right) \right) \otimes S^i \Pi \mathfrak{k}^* \ \ , \ \ V_{\mathfrak{k}}^{(m-p|n-q)} \defeq \bigoplus_{i=0}^\infty \left( S^i \Pi \mathfrak{k} \otimes \mathpzc{B}er \left( \mathfrak{k} \right) \right) \otimes S^i \Pi \mathfrak{h}^* \ .
\end{equation}
These modules correspond to pseudoforms associated with a given sub-structure $\mathfrak{h}$ and can be used to produce two pseudoform complexes. In \cite{CE2} the goal was the extension of Chevalley-Eilenberg cohomology to pseudoform complexes and the extension of Hochshild-Serre's spectral sequences construction and some interesting features were found. For example, because of the infinite dimensionality of pseudoform modules, some classical theorems of Lie algebra cohomology do not extend to pseudoforms, as the relation between cohomology classes and invariants (which, indeed, holds in the case of finite-dimensional modules).

The idea behind the construction of pseudoforms in the algebraic context is very simple: pseudoforms arise as integral forms of sub-structures. Notice that this was already suggested in \cite{Manin}, and in \cite{Witten}. Here, we highlight the main idea leaving its realisation to future developments: consider a supermanifold $\mathcal{SM}$, i.e., the data of a topological space $\left| \mathcal{SM} \right|$ and a sheaf of (supercommutative) rings of functions $\mathcal{O}_{\mathcal{SM}}$. Suppose we are given a submanifold $\mathcal{Y}$, we have then the data of a subset $\left| \mathcal{Y} \right| \subseteq \left| \mathcal{SM} \right|$ and its structure sheaf $\mathcal{O}_{\mathcal{Y}}$. We can define the sheaf of ideals $\mathcal{J}_{\mathcal{Y}}$ as the subsheaf of $\mathcal{O}_{\mathcal{SM}}$ consisting of those functions which vanish at $\mathcal{U} \cap \mathcal{Y}$, for all open sets $\mathcal{U}$. We then have the short exact sequence of sheaves
\begin{equation}
	0 \to \mathcal{J}_{\mathcal{Y}} \to \mathcal{O}_{\mathcal{SM}} \to \mathcal{O}_{\mathcal{Y}} \to 0 \ .
\end{equation}
One can introduce integral forms on $\mathcal{Y}$ and our idea is to interpret their extension as pseudoforms in $\mathcal{SM}$. On the level of PCOs, the idea is to define a PCO associated with the sheaf $\mathcal{J}_{\mathcal{Y}}$, in a similar way as the one described in the previous sections, and to define a PCO associated with the reduced manifold $\mathcal{Y}_{red}$ of $\mathcal{Y}$. The PCO defined for $\mathcal{SM}$ should be related to the chain composition of these two. In terms of the distributional realisation, this would mean to factorise (up to a sign) the PCO on $\mathcal{SM}$ as the product $\mathbb{Y}_{\mathcal{SM}} = \pm \mathbb{Y}_{\mathcal{Y} \hookrightarrow \mathcal{SM}} \wedge \mathbb{Y}_{\mathcal{Y}}$.

Up to now, it is not clear under which assumptions on $\mathcal{SM}$ and $\mathcal{Y}$ one can develop this argument; as stated at the beginning of the section, this is left for future developments.

\section*{Acknowledgements}
\noindent I am grateful to P.A. Grassi, B. Jur\v co, S. Noja and R. Noris for their help and their many suggestions on the draft. I am also grateful to A. Caddeo for various discussions on the physical relevance of PCOs. I acknowledge GA\v CR grant EXPRO 19-28628X for financial support.

\end{document}